\documentclass[graybox]{svmult}

\usepackage{mathptmx}       % selects Times Roman as basic font
\usepackage{helvet}         % selects Helvetica as sans-serif font
\usepackage{courier}        % selects Courier as typewriter font
\usepackage{type1cm}        % activate if the above 3 fonts are
\usepackage{makeidx}         % allows index generation
\usepackage{graphicx}        % standard LaTeX graphics tool
\usepackage{multicol}        % used for the two-column index
\usepackage[bottom]{footmisc}% places footnotes at page bottom

\usepackage{psfrag}

\makeindex             % used for the subject index
\begin{document}

\title*{Random Market Models with an $H$-Theorem}
\author{R. L\'opez-Ruiz, E. Shivanian and J.L. L\'opez}
\institute{Ricardo L\'opez-Ruiz \at Dept. of Computer Science \& BIFI, University of Zaragoza, Zaragoza, Spain \email{rilopez@unizar.es}
\and Elyas Shivanian \at Dept. of Mathematics, Imam Khomeini International University, Qazvin, Iran  \email{shivanian@ikiu.ac.ir}
\and Jos\'e Luis L\'opez \at Dept. of Math. Engineering, Public University of Navarre, Pamplona, Spain \email{jl.lopez@unavarra.es}}
\maketitle

\abstract*{In this communication, some economic models given by functional mappings are addressed.
These are models for random markets where agents trade by pairs and exchange their money in a 
random and conservative way.
They display the exponential wealth distribution as asymptotic equilibrium,
independently of the effectiveness of the transactions and of the limitation of the total wealth. 
The entropy increases with time in these models and the existence of an H-theorem is 
computationally checked. Also, it is shown that any small perturbation of the models equations
make them to lose the exponential distribution as an equilibrium solution. }

\abstract{In this communication, some economic models given by functional mappings are addressed.
These are models for random markets where agents trade by pairs and exchange their money in a 
random and conservative way.
They display the exponential wealth distribution as asymptotic equilibrium,
independently of the effectiveness of the transactions and of the limitation of the total wealth. 
The entropy increases with time in these models and the existence of an H-theorem is 
computationally checked. Also, it is shown that any small perturbation of the models equations 
make them to lose the exponential distribution as an equilibrium solution. }

\section{Introduction} 
\label{intro}

In the last years, it has been reported \cite{dragu2001,chakra2010} that in western societies,
around the 95\% of the population, the middle and lower economic classes of society arrange 
their incomes in an exponential wealth distribution. The incomes of the rest of the population, 
around the 5\% of individuals, fit a power law distribution \cite{solomom97}.

A kind of models considering the randomness associated to markets are the gas-like models \cite{yakoven2009}.
These random models interpret economic exchanges of money between agents similarly 
to collisions in a gas where particles share their energy \cite{dragu2000}.

In this communication, we consider a continuous version of an homogeneous gas-like model \cite{lopezruiz2011,lopez2012} 
which we generalize to a situation where the agents present a control parameter to decide the degree of interaction with
the rest of economic agents \cite{lopezruiz2013} and also to another new situation where there is an upper limit of 
the total richness.

The appearance of the exponential (Gibbs) distribution as a fixed point for all these three cases
is mathematically explained \cite{lopez2012}. Also, the increasing of the entropy when these systems evolve toward
the asymptotic equilibrium is checked. This is associated with the existence of an $H$-theorem for all these economic 
models \cite{shivanian2012,apenko2013}.
Despite their apparent simplicity, these models  based on functional mappings can help to enlighten the reasons 
of the ubiquity of the exponential distribution in many natural phenomena but in particular in the random markets.

\section{The continuous gas-like model}
\label{sec:1}

We consider an ensemble of economic agents trading with money by pairs in a random manner.
The discrete version of this model is as follows \cite{dragu2000}.
For each interacting pair $(m_i,m_j)$ of the ensemble of $N$ economic
agents the trading rules can be written as
\begin{eqnarray*}
m'_i &=& \epsilon \; (m_i + m_j), \nonumber\\
m'_j &=& (1 - \epsilon) (m_i + m_j), \label{model1}\\
i , j &=& 1 \ldots N, \nonumber
\end{eqnarray*}
where $\epsilon$ is a random number in the interval $(0,1)$.
The agents $(i,j)$ are randomly chosen. Their initial money $(m_i, m_j)$,
at time $t$, is transformed after the interaction in $(m'_i, m'_j)$ at time $t+1$.
The asymptotic distribution $p_f(m)$, obtained by numerical simulations,
is the exponential (Boltzmann-Gibbs) distribution,
\begin{displaymath}
p_f(m)=\beta \exp(-\beta \,m), \hspace{0.5cm}\hbox{with}\hspace{0.5cm}\beta={1/ <m>_{gas}}, 
\label{eq-exp}
\end{displaymath}
where $p_f(m) {\mathrm d}m$ denotes the PDF ({\it probability density function}), i.e.
the probability of finding an agent with money (or energy in a gas system) between
$m$ and $m + {\mathrm d}m$.
Evidently, this PDF is normalized, $\vert\vert p_f\vert\vert=\int_0^{\infty} p_f(m){\mathrm d}m=1$.
The mean value of the wealth, $<m>_{gas}$, can be easily calculated directly from the gas
by $<m>_{gas}=\sum_i m_i/N$.

The continuous version of this model \cite{lopezruiz2011} considers the evolution of
an initial wealth distribution $p_0(m)$ at each time step $n$ under the action of an operator $T$.
Thus, the system evolves from time $n$ to time $n+1$ to asymptotically
reach the equilibrium distribution $p_f(m)$, i.e.
\begin{displaymath}
\lim_{n\rightarrow\infty} {T}^n \left(p_0(m)\right) \rightarrow p_f(m). 
\label{eq-operT}
\end{displaymath}
In this particular case, $p_f(m)$ is the exponential distribution with the same
average value, $<p_f>$, than the initial one, $<p_0>$,
due to the local and total richness conservation.

The derivation of the operator $T$ is as follows \cite{lopezruiz2011}.
Suppose that $p_n$ is the wealth distribution in the ensemble at time $n$.
The probability to have a quantity of money $x$ at time $n+1$ will be the sum of the
probabilities of all those pairs of agents $(u,v)$ able to
produce the quantity $x$ after their interaction, that is, all the pairs verifying $u+v>x$.
Thus, the probability that two of these agents with money $(u,v)$ interact between them is
$p_n(u)*p_n(v)$. Their exchange is totally random and then they can give rise with equal
probability to any value $x$ comprised in the interval $(0,u+v)$. Therefore, the probability
to obtain a particular $x$ (with $x<u+v$) for the interacting pair $(u,v)$ will be
$p_n(u)*p_n(v)/(u+v)$.  Then, $T$ has the form of a nonlinear integral operator,
\begin{displaymath}
p_{n+1}(x)={T}p_n(x) = \int\!\!\int_{u+v>x}\,{p_n(u)p_n(v)\over u+v}
\; {\mathrm d}u{\mathrm d}v \,.   
\label{eq-T}
\end{displaymath}

If we suppose $T$ acting in the PDFs space, it has been proved \cite{lopez2012}
that $T$ conserves the mean wealth of the system, $<Tp>=<p>$. It also conserves
the norm ($\vert\vert \cdot\vert\vert$), i.e. $T$ maintains the total number of agents
of the system, $\vert\vert T p\vert\vert=\vert\vert p\vert\vert=1$, that
by extension implies the conservation of the total richness of the system.
We have also shown that the exponential distribution $p_f(x)$ with the right average value
is the only steady state of $T$, i.e. $T p_f=p_f$. Computations also seem to suggest
that other high period orbits do not exist.
In consequence, it can be argued that the convergence relation toward the limit point 
above explained is true. We sketch some of these properties.

First, in order to set up the adequate mathematical framework,
we provide the following definitions.

\begin{definition}
We introduce the space $L_1^+$ of positive functions (wealth distributions)
in the interval $[0,\infty)$,
$$
L_1^+[0,\infty)=\lbrace y:[0,\infty)\to R^+\cup\lbrace0\rbrace,
\hskip 2mm \vert\vert y\vert\vert<\infty\rbrace,
$$
with norm
$$
\vert\vert y\vert\vert=\int_0^\infty y(x) dx.
$$
\end{definition}

\begin{definition}
\label{def-mean1}
We define the mean richness $<x>_y$ associated to a wealth distribution $y\in L_1^+[0,\infty)$ as
the mean value of $x$ for the distribution $y$. Then,
$$
<x>_y = \vert\vert xy(x)\vert\vert=\int_0^\infty xy(x) dx.
$$
\end{definition}

\begin{definition}
\label{def-region}
For $x\ge 0$ and $y\in L_1^+[0,\infty)$ the action of operator $T$ on $y$ is defined by
$$
T(y(x))=\int\int_{S(x)} dudv{y(u)y(v)\over u+v},
$$
with $S(x)$ the region of the plane representing the pairs of agents $(u,v)$ which can
generate a richness $x$ after their trading, i.e.
$$
S(x)=\lbrace (u,v), \hskip 2mm u,v>0,\hskip 2mm u+v>x\rbrace.
$$
\end{definition}

Now, we remind the following results recently presented in Ref. \cite{lopez2012,lopezruiz2013}.

\begin{theorem}
\label{teorema-norma}
For any $y\in L_1^+[0,\infty)$ we have that
$\vert\vert Ty\vert\vert=\vert\vert y\vert\vert^2$.
In particular, consider the subset of PDFs in $L_1^+[0,\infty)$, i.e. the unit sphere
$B=\lbrace y\in L_1^+[0,\infty)$, $\vert\vert y\vert\vert=1\rbrace$. Observe that
if $y\in B$ then $Ty\in B$. (It means that the number of agents in the economic system is conserved in time).
\end{theorem}

\begin{theorem}
The mean value $<x>_y$ of a PDF $y$ is conserved, that is $<x>_{Ty}=<x>_y$ for any $y\in B$.
(It means that the mean wealth, and by extension the total richness, of the economic
system are preserved in time).
\end{theorem}

\begin{theorem}
Apart from $y=0$, the one-parameter family of functions
$y_\alpha(x)= \alpha e^{-\alpha x}$, $\alpha>0$, are the unique fixed points
of $T$ in the space $L_1^+[0,\infty)$.
\label{teorema-unicidad}
\end{theorem}

\begin{proposition} 
For some members $y,w\in B$,
        $\vert\vert Ty-Tw\vert\vert\geq\vert\vert y-w\vert\vert$, hence $T$ is not a contraction.
\end{proposition}

\begin{example}
Take $y(x)={1\over(1+x)^2}$ and $w(x)=e^{-x}$ which belong to $B$. By using Mathematica,
it is seen that $\vert\vert y-w\vert\vert=0.407264$ and $\vert\vert Ty-Tw\vert\vert=0.505669$.
\end{example}

If we consider the restriction of $T$ for the subset $B_{x_0}$ of distributions with the same mean wealth $x_0$,
i.e. $B_{x_0}=\{y\in B| <x>_y=x_0\}$, then by using the Laplace transform of the operator $T$,
it has been proved \cite{gutriel2013} that $T$ is a contraction in  $B_{x_0}$, 
hence the truth of the following relation:
$$
\lim_{n\rightarrow\infty} {T}^ny(x)=\left\{\begin{array}{lcl}
\delta e^{-\delta x} & & with \hskip 5mm \delta=1/x_0\,,\\
& or & \\
0^+ & & when \hskip 3mm <x>_y=+\infty\,.
\end{array}\right.
$$

\vskip 0.2cm
Let us observe that the above pointwise limit of $T^ny$ when $n\to\infty$ can be outside of $B$
in the case that $<x>_y=+\infty$. See the next example.

\begin{example}
Take $y(x)={1\over(1+x)^2}$ which belongs to $B$, with $<x>_y=+\infty$.
Evidently, $T^ny\in B$ for all $n$. But it can be seen that $\lim_{n\rightarrow\infty} {T}^ny(x)=0^+\notin B$.
\label{example-pareto}
\end{example}

\begin{example}
Assume now the rectangular distribution:
$y(x)={1\over 2}$ if $2<x<4$, and $y(x)=0$ otherwise.
So, $y\in B$ and $\delta={1\over 3}$,
then the steady state in this case is $\mu(x)={1\over 3}e^{-{1\over 3}x}$. 
We find numerically that $\vert\vert y-\mu\vert\vert > \vert\vert Ty-\mu\vert\vert 
> \vert\vert T^2y-\mu\vert\vert > \vert\vert T^3y-\mu\vert\vert$, and so on.
It is shown in Fig. \ref{fig-y-rectangular}.
Then we can guess that $\lim_{n\rightarrow\infty}\vert\vert T^ny-\mu\vert\vert=0$.
\end{example}

\begin{figure}[h]
\begin{center}
\psfrag{B}{} \psfrag{A}{\large\scriptsize (a)}
\includegraphics[width=1.5in, height=1.3in]{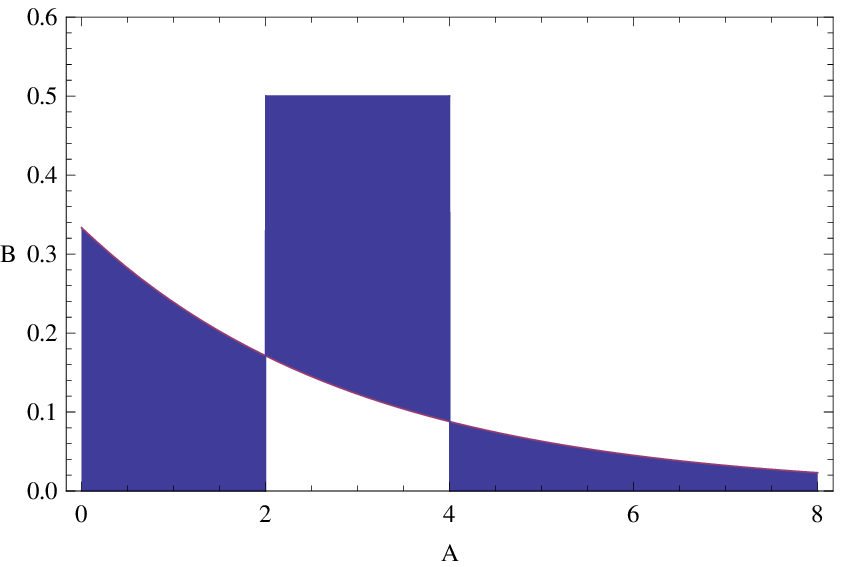} \hskip 2 mm
\psfrag{B}{} \psfrag{A}{\large\scriptsize (b)}
\includegraphics[width=1.5in, height=1.3in]{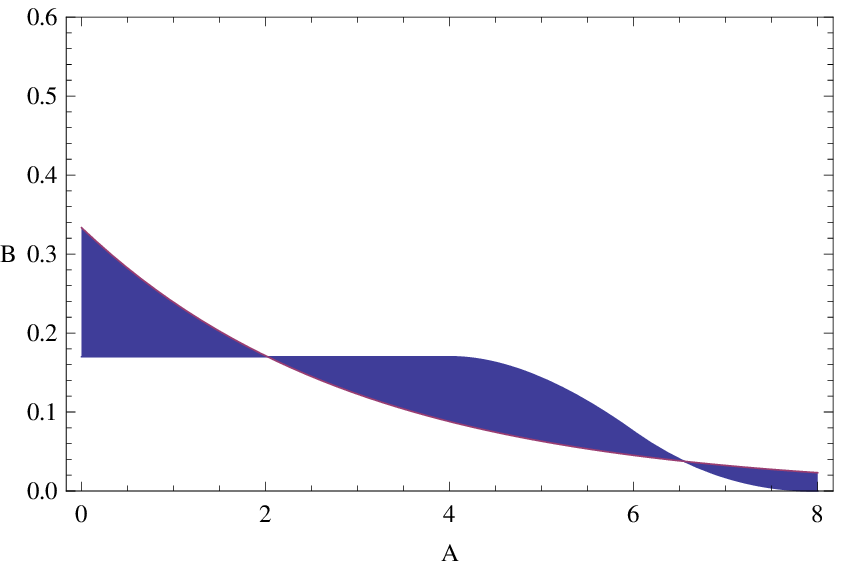} \hskip 2 mm
\psfrag{B}{} \psfrag{A}{\large\scriptsize (c)}
\includegraphics[width=1.5in, height=1.3in]{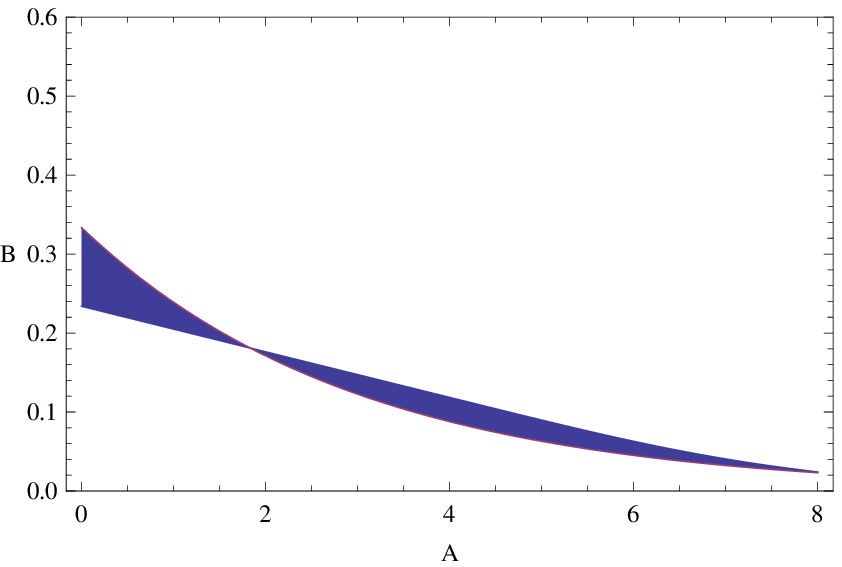}
\caption{Plot of $y(x)={1\over 2}$ if $2<x<4$, and $y(x)=0$ otherwise,
$T$-iterates of $y$ and $\mu(x)={1\over 3}e^{-{1\over 3}x}$.
(a) $\vert\vert y-\mu\vert\vert$, (b) $\vert\vert Ty-\mu\vert\vert$,
(c) $\vert\vert T^2y-\mu\vert\vert$.}
\label{fig-y-rectangular}
\end{center}
\end{figure}

If we consider the entropy of $y(x)$ given by $H=-\int y(x)\log y(x)dx$, then it is found 
that $H$ increases in a monotonic way when $T$ is successively applied to an initial state $y_0(x)$. 
If we define $H_n=H(T^ny_0(x))$, then the following $H$-theorem \cite{apenko2013} yields:
$$
\lim_{n\rightarrow\infty} H_n = H(\delta e^{-\delta x}) \hskip 0.5cm \hbox{with} \hskip 5mm \delta=1/<x>_{y_0}\,, 
$$
$$
\hskip 0.3cm  H_n\leq H_{n+1}  \hskip 5mm \forall n \,. \hskip 3cm
$$

Summarizing, the system has a fixed point, $\delta e^{-\delta x}$, which is
asymptotically reached depending on the initial average value $<x>_{y_0}$ and
following a trajectory of increasing entropy. This behavior is essentially maintained 
in the extension of this model for other similar random markets (see the next sections).

\section{The continuous gas-like model with homogeneous effectiveness}

Let us think now that many of the economical transactions planned in markets are
not successful and they are finally frustrated. It means that markets are not totally
effective. We can reflect this fact in our model in a qualitative way by defining a
parameter $\lambda\in [0,1]$ which indicates the {\it degree of effectiveness} of the
random market. When $\lambda=1$ the market will have total effectiveness and all the
operations will be performed under the action of the random rules (\ref{model1}).
The evolution of the system in this case is given by the operator $T$.
When $\lambda=0$, all the operations
become frustrated, there is no exchange of money between the agents and then
the market stays frozen in its original state. The operator representing this type
of dynamics is just the identity operator. Therefore, we can establish a {\it generalized
continuous economic model} whose evolution in the PDFs space is determined by the
operator $T_{\lambda}$, which depends on the parameter $\lambda$ as follows:

\begin{definition}
$T_{\lambda}y(x)=(1-\lambda)y(x)+\lambda Ty(x)$, with $\lambda\in[0,1]$.
\end{definition}

Observe that the parameter $(1-\lambda)$ can also be interpreted as
a kind of {\it saving propensity} of the agents, in such a way that for $\lambda=1$
they do not save anything and they game all their resources,
and for $\lambda=0$ they save the totality of their money and
then all the transactions are frustrated and the market stays in a frozen state.

We present some properties of the operator $T_{\lambda}$,
which shows a dynamical behavior essentially similar to the behavior of $T$.
Concretely, the exponential distribution is also
the asymptotic wealth distribution reached by the system governed by $T_{\lambda}$,
independently of the effectiveness $\lambda$ of the random market.

Let us observe that
$T_{\lambda}=I$ for $\lambda=0$ and $T_{\lambda}=T$ for $\lambda=1$, where $I$
is the identity operator.

\begin{proposition}
$T_{\lambda}$ conserves the norm, i.e. for each $y\in B$, we have $T_{\lambda}y\in B$.
\end{proposition}

\begin{proposition}
$T_{\lambda}$ conserves the average value of $y\in B$, i.e. $<x>_y=<x>_{T_{\lambda}y}$,
where $<x>_y$ represents the mean value expressed in Definition \ref{def-mean1}.
\end{proposition}

\begin{theorem}
For any $\lambda\in (0,1)$, the operators $T$ and $T_{\lambda}$ have the same fixed points.
\label{theor-same}
\end{theorem}

\begin{corollary}
The function $y(x)=0$ and the family of exponential distributions
$y_{\delta}(x)=\delta e^{-\delta x}$, $\delta>0$, are the only fixed points
of $T_{\lambda}$ in $L_1^+[0,\infty)$, with $\lambda\in (0,1]$.
\end{corollary}

\begin{theorem}
Suppose that for a given $\lambda\in (0,1)$ we have 
$\lim_{n\to\infty}\vert\vert T_{\lambda}^ny(x)-\mu(x)\vert\vert=0$, with $\mu(x)$
a continuous function, then $\mu(x)$ should be the fixed point of the operator 
$T_{\lambda}$ for the initial condition $y(x)\in B$. In other words,
$\mu(x)=\delta e^{-\delta x}$ with $\delta={1 \over <x>_y}$.
\end{theorem}

\begin{example}
Take the Gamma distribution $y(x)=xe^{-x}$, so that $y\in
B$ and $\delta={1\over 2}$, then in this case $\mu(x)={1\over
2}e^{-{1\over 2}x}$. For $\lambda= 0.5$, we find numerically that
$\vert\vert y-\mu\vert\vert=0.368226$, $\vert\vert
T_{\lambda}y-\mu\vert\vert=0.273011$, $\vert\vert
T_{\lambda}^2y-\mu\vert\vert=0.206554$, $\vert\vert
T_{\lambda}^3y-\mu\vert\vert=0.158701$, and so on. It is shown in
Fig. \ref{fig-y-gamma1}. Then we can guess that
$\lim_{n\rightarrow\infty}\vert\vert
T_{\lambda}^ny-\mu\vert\vert=0$.
\end{example}

\begin{figure}[h]
\begin{center}
\psfrag{B}{} \psfrag{A}{\large\scriptsize (a)}
\includegraphics[width=1.5in, height=1.3in]{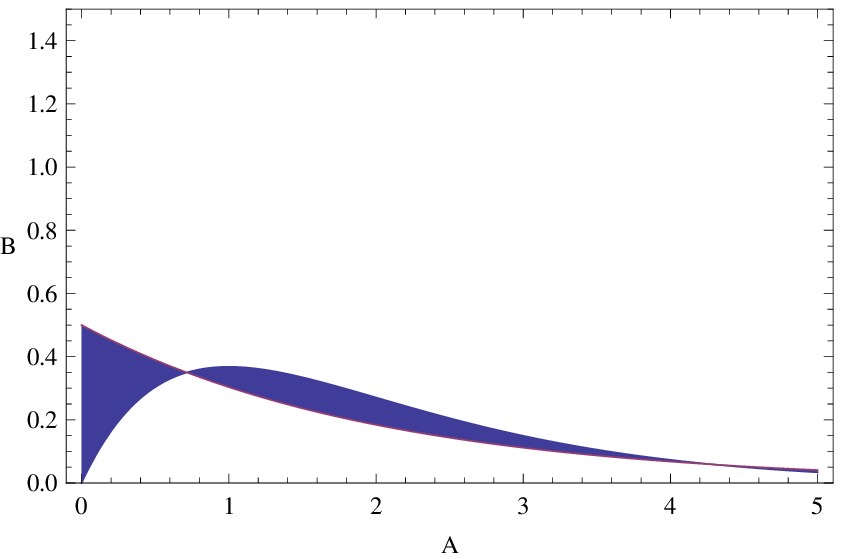} \hskip 2 mm
\psfrag{B}{} \psfrag{A}{\large\scriptsize (b)}
\includegraphics[width=1.5in, height=1.3in]{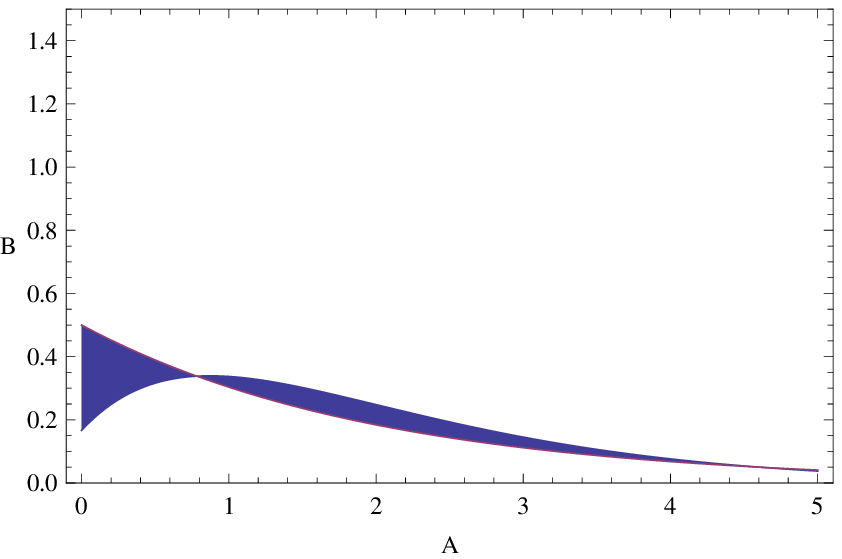} \hskip 2 mm
\psfrag{B}{} \psfrag{A}{\large\scriptsize (c)}
\includegraphics[width=1.5in, height=1.3in]{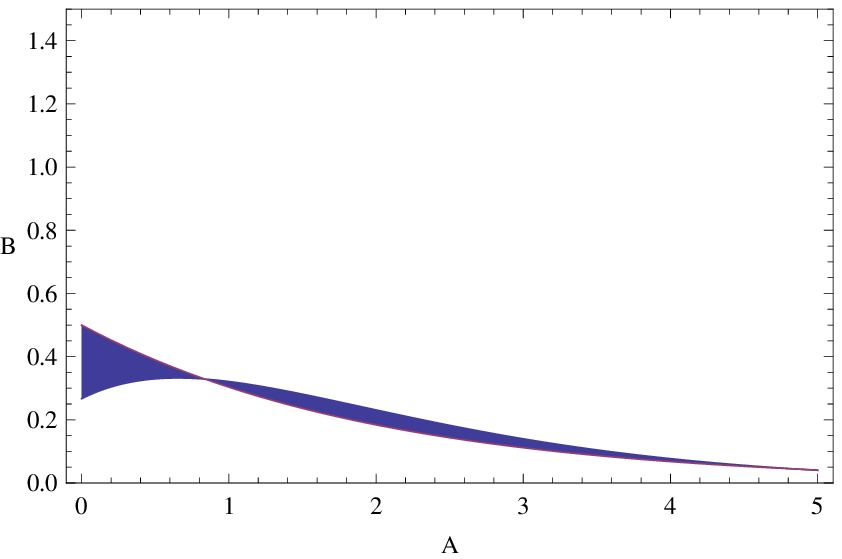}
\caption{Plot of $y(x)=xe^{-x}$, $T_{\lambda}$-iterates of $y$
for $\lambda=0.5$ and $\mu(x)={1\over 2}e^{-{1\over 2}x}$. (a)
$\vert\vert y-\mu\vert\vert$, (b) $\vert\vert
T_{\lambda}y-\mu\vert\vert$, (c) $\vert\vert
T_{\lambda}^2y-\mu\vert\vert$.} \label{fig-y-gamma1}
\end{center}
\end{figure}

Also, the increasing of entropy in the system evolution can be checked.
Then, similarly to the first model, this model has a fixed point, $\delta e^{-\delta x}$, 
which is asymptotically reached depending on the initial average value $<x>_{y_0}$ and
following a trajectory of increasing entropy. The difference with the first model remains 
in the transient towards equilibrium, that evidently is a longer time for a lower effectiveness 
$\lambda$ of the random market.

%%% LIMITACION DE RIQUEZA

\section{The continuous gas-like model with limitation of the richness}

Here, we study the effect of a limitation in the maximum richness that an agent can have.
We establish this upper limit to be $\Lambda$ for $x$: $x\in[0,\Lambda]$.
Now, the mean wealth of the system is:
$$
<x>_y=\int_0^\Lambda xy(x)dx.
$$
Evidently, $<x>_y<\Lambda$. 

The existence of the cut-off $\Lambda$ in the economic system not only means
that agents can not have more money than $\Lambda$ if not the creation of some kind of control 
on the system that do not let the agents to perform trades that surpass the limit $\Lambda$.
In the case of interaction by pairs, it implies that only the pairs verifying $u+v<\Lambda$ are allowed to trade
and then they can exchange their money according to rules (\ref{model1}). In the rest of interactions surpassing 
the upper limit, that is $u+v>\Lambda$, the agents are not allowed to trade and then they conserve their original money.  
Hence, the generalization of the operator $T$ for this system is the following:
$$
[T_\Lambda y](x)=\int \int_{x\le u+v\le\Lambda}{y(u)y(v)\over u+v}dudv+y(x)\int_{\Lambda-x}^\Lambda y(v)dv,
$$
where the first term integrates the allowed trades according to rules (\ref{model1}) and the second term gives 
account of the total probability of encounters with forbidden trades that an agent of richness $x$ can have with other 
agents of the ensemble. 

Observe that $\lim_{\Lambda\to\infty}T_\Lambda=T$. Also:
$$
[T_\Lambda y](x)=\int_0^x du \int_{x-u}^{\Lambda-u}dv{y(u)y(v)\over u+v}+ 
$$
$$
+\int_x^\Lambda du \int_0^{\Lambda-u}dv{y(u)y(v)\over u+v}+y(x)\int_{\Lambda-x}^\Lambda y(v)dv.
$$

\begin{theorem} 
For any $y\in L_1^+[0,\infty)$ and $\Lambda>0$ we have that $\vert\vert T_\Lambda y\vert\vert=\vert\vert y\vert\vert^2$. 
In particular, if $\vert\vert y\vert\vert=1$ then $\vert\vert T_\Lambda y\vert\vert=1$.
\end{theorem}

\begin{proof}
\smartqed
Take $y\in L_1^+[0,\infty)$. Then 
\begin{eqnarray*}
& \vert\vert T_\Lambda y\vert\vert= \int_0^\Lambda[T_\Lambda y](x)dx = \\
& \int_0^\Lambda du\int_0^{\Lambda-u}dv\int_0^{u+v}dx{y(u)y(v)\over u+v}
+\int_0^\Lambda y(x)dx\int_{\Lambda-x}^\Lambda y(v)dv= \\
& \int_0^\Lambda y(u)du\int_0^{\Lambda-u}y(v)dv+\int_0^\Lambda y(u)du\int_{\Lambda-u}^\Lambda y(v)dv= \vert\vert y\vert\vert^2. 
\end{eqnarray*}
\qed
\end{proof}

\begin{theorem}
The mean richness is conserved by $T_\Lambda$, that is $<x>_{T_\Lambda y}=<x>_y$ for any $y\in B$.
\end{theorem}

\begin{proof}
\smartqed
\begin{eqnarray*}
& <x>_{T_\Lambda y}= \int_0^\Lambda xy(x)dx= \\
& \int_0^\Lambda du\int_0^{\Lambda-u}dv\int_0^{u+v}xdx{y(u)y(v)\over u+v}
+\int_0^\Lambda xy(x)dx\int_{\Lambda-x}^\Lambda y(v)dv= \\
& {1\over 2}\int_0^\Lambda du\int_0^{\Lambda-u}dv(u+v)y(u)y(v)+\int_0^\Lambda xy(x)dx\int_{\Lambda-x}^\Lambda y(v)dv= \\
& {1\over 2}\int_0^\Lambda uy(u)du\int_0^{\Lambda-u}y(v)dv+{2\over 2}\int_0^\Lambda uy(u)du\int_{\Lambda-u}^\Lambda y(v)dv
+{1\over 2}\int_0^\Lambda y(u)du\int_0^{\Lambda-u}vy(v)dv= \\ 
& {1\over 2}\int_0^\Lambda uy(u)du\int_0^\Lambda y(v)dv+{1\over 2}\int_0^\Lambda y(v)dv\int_{\Lambda-v}^\Lambda uy(u)du
+{1\over 2}\int_0^\Lambda y(u)du\int_0^{\Lambda-u}vy(v)dv = \\ 
& {1\over 2}\int_0^\Lambda uy(u)du\int_0^\Lambda y(v)dv+{1\over 2}\int_0^\Lambda y(v)dv\int_0^\Lambda uy(u)du=
{1\over 2}<x>_y+{1\over 2}<x>_y=<x>_y.
\end{eqnarray*}
\qed
\end{proof}

\begin{theorem} 
The function
$$
y_{a,\Lambda}(x)={ae^{-ax}\over 1-e^{-a\Lambda}}
$$
has $\vert\vert y_{a,\Lambda}\vert\vert=1$ and is a fixed point of the operator $T_\Lambda$ for any $a>0$. 
The mean richness for this function is
$$
<x>_{y_{a,\Lambda}}={1\over a}+{\Lambda\over 1-e^{a\Lambda}}.
$$
\end{theorem}

\begin{proof}
\smartqed
It is just a straightforward computation.
\qed
\end{proof}

\vskip 0.2cm
\begin{proposition}
For the fixed point $y_{a,\Lambda}(x)$, we have $2<x>_{y_{a,\Lambda}}<\Lambda$.
\end{proposition}

Hence, if we define $m=<x>_{y_{a,\Lambda}}$, we can consider {\it the middle class}, $CM$, 
as all those agents having richness between $m/2$ and $2m$, that is,
$$
CM(a,\Lambda)=\int_{m/2}^{2m}y_{a,\Lambda}(x)dx={e^{-am/2}-e^{-2am}\over 1-e^{-a\Lambda}}.
$$
The richness accumulated by the middle class is:
$$
<xCM>(a,\Lambda)=\int_{m/2}^{2m}xy_{a,\Lambda}(x)dx=
$$
$$
=m{(2+am)e^{-am/2}-2(1+2am)e^{-2am}\over 2am[1-e^{-a\Lambda}]},
$$
where
$$
am=1+{a\Lambda\over 1-e^{a\Lambda}}, \hskip 2cm x=a\Lambda.
$$
When we plot $CM(x)$ or $<xCM>(x)$ for fixed $m$, we see that it is always a decreasing function of $x=a\Lambda$. 
Therefore, the smaller the richness limit $\Lambda$ is, the larger the middle class is.

The same tendency can be observed if we calculate the mean wealth per individual of the middle class:
$$
{<xCM>(a,\Lambda)\over <CM>(a,\Lambda)}={m\over am}+{m[e^{-am/2}-4e^{-2am}]\over 2[e^{-am/2}-e^{-2am}]},
$$
which is also a decreasing function of $x=a\Lambda$, for fixed $m$.

The proportion of the total richness accumulated by the middle class is:
$$
{<xCM>(a,\Lambda)\over m}={(2+am)e^{-am/2}-2(1+2am)e^{-2am}\over 2am[1-e^{-a\Lambda}]},
$$
that is again a decreasing function of $x=a\Lambda$.

Summarizing, an upper limit in the richness allowed  in a random market provokes an enlargement 
of the middle class and also an enrichment of such a middle class.

%% PERTURBATION OF THE MODELS

\section{Slightly perturbed gas-like models}

Here, we put in evidence that the asymptotic equilibrium distributions for these models are not stable
under slight perturbations. It means that even an small modification, that conserves the mean value 
and the total wealth of the system, provokes the loss of the exponential distribution as a fixed point 
of the perturbed model equations. This fact also has its consequences on the behavior of the
entropy of the system, concretely the $H$-theorem is not already verified, unless new forms for the $H$
functional are introduced.

As an example, define the modified operator
$$
(T_Ky)(x)=\int\int_{u+v\ge x}K(u,v,x){y(u)y(v)\over u+v}dudv, 
$$
with the kernel
$$
K(u,v,x)=\sum_{n=0}^N(n+1)a_n\left({x\over u+v}\right)^n,
$$
where eventually $N$ may be infinity. It is straightforward to check that
$$
\vert\vert Ty\vert\vert=\vert\vert y\vert\vert^2\sum_{n=0}^N a_n
$$
and
$$
<x>_{T_Ky}=2<x>_y\sum_{n=0}^N {n+1\over n+2}a_n.
$$
Therefore, the operator verifies $T_K: B\to B$ and conserves the wealth when

\begin{eqnarray*}
\sum_{n=0}^N a_n= & 1, \\ 
\sum_{n=0}^N {n+1\over n+2}a_n= & {1\over 2}. 
\end{eqnarray*}

For $N=1$ the unique solution of this system is $a_0=1$ and $a_1=0$, that is, the well known operator studied elsewhere. 
For $N=2$ we have an infinity of solutions parametrized with $\epsilon\in R$: $a_0=1-\epsilon/3$, $a_1=\epsilon$ and 
$a_2=-2\epsilon/3$. That is,
$$
K(u,v,x)=1-{\epsilon\over 3}+{2\epsilon x\over u+v}-{2\epsilon x^2\over (u+v)^2}.
$$

If we take the exponential distribution $y(x)=a e^{-ax}$ and $N=2$ we find that
$$
(T_Ky)(x)=\left(1-{\epsilon\over 3}\right)ae^{-ax}-2\epsilon xa^2 Ei(-ax)-2\epsilon x^2a^2\left[{e^{-ax}\over x}+a Ei(-ax)\right],
$$
where $Ei(x)$ is the Exponential Integral,
$$
Ei(-ax)=-\int_x^\infty{e^{-at}\over t}dt.
$$
It means that $y(x)=a e^{-ax}$ is not the fixed point of the perturbed operator $T_K$ and a new asymptotic equilibrium emerges for
this $\epsilon$-slightly modified system.

\section{Conclusions}

Different versions of a continuous economic model \cite{lopezruiz2011} that takes 
into account idealistic characteristics of the markets have been considered. In these models,
the agents interact by pairs and exchange their money in a random way.
The asymptotic steady state of these models is the exponential wealth distribution.
The system decays to this final distribution with a monotonic increasing of the entropy
taking its maximum value just on the equilibrium. These are specific $H$-theorems
that can be computationally checked, independently on the effectiveness of the markets
or the limitation of the richness in the economic system. Also, it has been argued that
slight modifications of these of models provoke the loss of the exponential distribution
as an asymptotic equilibrium and its correspondent consequences for the establishment 
of an $H$-theorem for the new perturbed models.

\end{document}